\documentclass[twoside]{article}

\usepackage{auxiliary/PRIMEarxiv}
\usepackage{auxiliary/arxivpackages}
\usepackage{auxiliary/mymacros}
\usepackage{auxiliary/commands}
\renewcommand*{\backrefalt}[4]{%
    \ifcase #1 \footnotesize{(Not cited.)}%
    \or        \footnotesize{(Cited on page~#2.)}%
    \else      \footnotesize{(Cited on pages~#2.)}%
    \fi}

\renewcommand{\cite}[1]{\citep{#1}}
\setcitestyle{authoryear,round,citesep={;},aysep={,},yysep={;}}

\newcommand{\mytitle}[1][]{}
\newcommand{\mysubtitle}[1][]{}
\newcommand{\myauthors}[1][]{Jiayi Zhao}
\newcommand{\myauthorsshort}[1][]{Zhao}

\title{\mytitle Fisher Markets with Social Influence
}

\author{
  Jiayi Zhao \\
 Department of Computer Science\\
 Pomona College \\
  Claremont, CA\\
  \texttt{jzae2019@mymail.pomona.edu} \\
   \And
  Denizalp Goktas, Amy Greenwald \\
  Department of Computer Science\\
  Brown University \\
  Providence, RI\\
  \texttt{ \{denizalp\_goktas,amy\_greenwald\}@brown.edu} \\
}

\begin{document}
\maketitle

\begin{abstract}
A Fisher market is an economic model of buyer and seller interactions in which each buyer's utility depends only on the bundle of goods she obtains.
Many people's interests, however, are affected by their social interactions with others.
In this paper, we introduce a generalization of Fisher markets, namely influence Fisher markets, which captures the impact of social influence on buyers' utilities.
We show that competitive equilibria in influence Fisher markets correspond to generalized Nash equilibria in an associated pseudo-game, which implies the existence of competitive equilibria in all influence Fisher markets with continuous and concave utility functions.
We then construct a monotone pseudo-game, whose variational equilibria and their duals together characterize competitive equilibria in influence Fisher markets with continuous, jointly concave, and homogeneous utility functions.
This observation implies that competitive equilibria in these markets can be computed in polynomial time under standard smoothness assumptions on the utility functions.
The dual of this second pseudo-game enables us to interpret the competitive equilibria of influence CCH Fisher markets as the solutions to a system of simultaneous Stackelberg games.
Finally, we derive a novel first-order method that solves this Stackelberg system in polynomial time, prove that it is equivalent to computing competitive equilibrium prices via \emph{t\^{a}tonnement}, and run experiments that confirm our theoretical results.
\end{abstract}

\section{Introduction}

The branch of mathematical economics that attempts to explain the behavioral relationship among supply, demand, and prices via equilibria dates back to the work of French economist 
\citep{walras}, and today is known as general equilibrium theory \citep{mas-colell}.
One of the seminal achievements in this area is the proof of existence of competitive equilibrium prices in Arrow-Debreu markets (\citeyear{arrow1954existence}).
In such a market, traders seek to ``purchase'' goods from others, by exchanging a part of their endowment of goods for various other goods.
A competitive equilibrium comprises an allocation of goods to traders together with good prices such that traders maximize their preferences over goods while ensuring that their spending does not exceed the value of their endowment, 
and the market clears: 
i.e., no more goods are allocated than the market supply and Walras' law holds, meaning the value of demand equals the value of supply.




In much of mainstream consumer theory \citep{mas-colell}, and in Arrow-Debreu markets, each trader's preference depends only on its own consumption.
Such models fail to capture the influence of social interactions on traders' interests. 
For example, the more friends one has who own an iPhone, the more one might prefer an iPhone.
Likewise, if a celebrity, e.g., Beyonce, wears a particular brand of bag, e.g., Telfar, then one's preference for that brand of bag might increase.
In an age of densifying social networks, it is becoming more and more essential that our economic models capture the effects of social interactions on individuals' preferences. 

To try to better understand the implications of social networks on market equilibria, \citet{Chen2011MakretwithSocialInfluence} recently proposed an extension of the Arrow-Debreu market model in which each trader's preference is influenced by the goods her neighbors obtains: \mydef{the Arrow-Debreu market with social influence}.
Formally, \citeauthor{Chen2011MakretwithSocialInfluence}'s model augments an Arrow-Debreu market with a social network connecting the traders, and then embeds this network's structure in each trader's utility function, thus inducing a preference relation over allocations of goods that depends both on the trader's and its neighbors' allocations.
The authors then study a modest generalization of competitive equilibrium in which traders maximize their utility, assuming the allocations of the other traders in the market, including their neighbors, are fixed. 

\citeauthor{Chen2011MakretwithSocialInfluence} analyze their model under two specific types of utilities: linear and threshold influence functions.
They prove existence of competitive equilibrium in their setting, when the graph underlying the economy is strongly connected and the utility functions' parameters guarantee non-satiation of the preferences they represent.
Under additional assumptions on the topology of the network, they also provide polynomial-time methods for computing competitive equilibria.

As the computation of competitive equilibrium in Arrow-Debreu markets is believed to be intractable, i.e., it is PPAD-complete \citep{chen2006settling, chen2009spending}, it seems unlikely that we can obtain positive computational results in such a broad setting.
In the last two decades, however, \mydef{Fisher markets} have emerged as an interesting special case of Arrow-Debreu markets in which competitive equilibria can be efficiently computed.
The Fisher market is a one-sided Arrow-Debreu market comprising one seller and multiple buyers, the latter of whom are endowed with an artificial currency called their budget,
rather than an endowment of goods. 

During the last two decades, a wide array of polynomial-time computability results have been established for Fisher markets \citep{devanur2002market, jain2005market, gao2020polygm, goktas2021minmax}.
One of the most interesting findings is the observation that the primal and dual solutions, respectively, to the \mydef{Eisenberg-Gale convex program} \cite{eisenberg1959consensus}, constitute competitive equilibrium allocations and prices in Fisher markets, and are computable in polynomial time assuming buyers with continuous, concave, and homogeneous utility functions representing locally non-satiated preferences \citep{devanur2002market, devanur2008market, jain2005market}. 
Moreover, \citeauthor{fisher-tatonnement}
show that solving the dual of the Eisenberg-Gale program via (sub)gradient descent amounts to solving the market via \mydef{\emph{t\^atonnement}}, an economic price-adjustment process dating back to \citeauthor{walras} (\citeyear{walras}), in which a fictional auctioneer increases (resp.\@ decreases) the prices of goods that are overdemanded (resp.\@ underdemanded) \cite{fisher-tatonnement}. 
Furthermore, \citet{goktas2021minmax} show that the dual of the Eisenberg-Gale program corresponds to a zero-sum Stackelberg game, in which \emph{t\^atonnement\/} surfaces as a no-regret learning dynamic for the auctioneer \cite{goktas2022robust}.

With the aim of obtaining stronger results on the existence and computation of competitive equilibrium in markets with social influence, we introduce a special case of the Arrow-Debreu market with social influence and a generalization of Fisher markets \citep{brainard2000compute}, which we call \mydef{Fisher markets with social influence}, or \mydef{influence Fisher markets} for short.
An influence Fisher market, as the name suggests, is a Fisher market in which buyers' utility functions depend not only on their own allocation, but also on their neighbors'.
In this paper, we provide existence and polynomial-time computability results for competitive equilibrium in influence Fisher markets.
We first extend \citeauthor{arrow1954existence}'s competitive equilibrium existence argument using their theory of pseudo-games 
to prove that a competitive equilibrium exists in all influence Fisher markets with continuous utility functions that are concave in each buyer's allocation. Contrary to \citeauthor{Chen2011MakretwithSocialInfluence}, our existence result makes no assumptions about the topology of the network.

Next, for all influence Fisher markets with continuous and homogeneous utilities, we construct a  similar pseudo-game with jointly convex constraints whose variational equilibria 
correspond to competitive equilibrium allocations.
This pseudo-game is monotone assuming the buyers' utility functions are jointly-concave); thus, we can solve for its variational equilibria as a variational inequality problem \citep{facchinei2007vi}.
This approach yields a polynomial-time algorithm that computes competitive equilibrium allocations in influence Fisher markets.
Moreover, as the pseudo-game comprises $\numbuyers$ different optimization problems, one per buyer, there are correspondingly $\numbuyers$ duals.
Surprisingly, 
the solutions to all of these duals yield the same competitive equilibrium prices!

{
Finally, following \citet{goktas2021minmax}, who reformulate the dual of the Eisenberg-Gale program as a zero-sum Stackelberg game, we likewise reformulate the $\numbuyers$ duals of our pseudo-game as a system of $\numbuyers$ simultaneous zero-sum Stackelberg games.
In \citeauthor{goktas2021minmax}'s dual, the leader is a fictitious auctioneer who sets prices, while the followers are a set of buyers who effectively play as a team; in our $n$ duals, each leader is again a fictitious auctioneer, but each follower is an individual buyer who best responds to the auctioneer's prices, \emph{given the other buyers' allocations}.
Thus, the buyers in this system play a Nash equilibrium.
Also following \citeauthor{goktas2021minmax}, we show that running subgradient descent on each leader's value function, i.e., the leader's utility function assuming the follower best-responds,
amounts to solving the market via \emph{t\^atonnement\/} in polynomial-time, as in (standard) Fisher markets.
The main difference between our algorithm and 
theirs
is that ours requires a Nash-equilibrium oracle, so that, given prices \emph{and the other buyers' allocations}, buyers can play best responses to one another.}


\amy{i don't know exactly how the FFP algo works. but it feels like there is opportunity for a sort-of multi-player envelope theorem, where the leader could run a nested gradient descent algo computing its derivative as a function of ALL the players' eqm actions.} \sadie{We tried the current envelope theorem for NE and believed that it would not provide us with the correct gradient for the aggregate value function (since it would take the derivative for all buyers' allocations at the same time, which ruins the thing). Instead, we use the old envelope theorem on each value function, treating other buyers' allocations as constant and obtain desired results.} \amy{what do you mean by the ``current'' envelope theorem?} \sadie{There is an Envelope theorem for NE in the paper "The Envelope Theorem and Comparative Statics of Nash Equilibria"}

\paragraph{Related Work} 
\citet{gao2020polygm} studied an alternative family of first-order methods for solving Fisher markets (only; not min-max Stackelberg games more generally), assuming linear, quasilinear, and Leontief utilities; such methods can be more efficient when markets are large.

Following \citeauthor{arrow1954existence}'s introduction of GNE, \citet{rosen1965gne} initiated the study of the mathematical and computational properties of GNE in pseudo-games with jointly convex constraints, proposing a projected gradient method to compute GNE.
Thirty years later, 
\citet{uryas1994relax} developed the first relaxation methods for finding GNEs, which were improved upon in subsequent works \citep{Krawczyk2000relax, von2009relax}.
Two other types of algorithms were also introduced to the literature: Newton-style methods \citep{facchinei2009generalized, dreves2017computing, von2012newton, izmailov2014error, fischer2016globally, dreves2013newton} and interior-point potential 
methods \citep{dreves2013newton}.
Many of these approaches are based on minimizing the exploitability of the pseudo-game, 
but others use variational inequality \citep{facchinei2007vi, nabetani2011vi} and Lemke methods \citep{Schiro2013lemke}.
Recently, this literature has established convergence guarantees for exploitability minimization \cite{goktas2022exploitability} and relaxation \cite{jordan2023first} methods.

\section{Preliminaries}

In this section, we define our main modeling tool, pseudo-games, and 
we introduce our object of study, Fisher markets with social influence, as a particular pseudo-game.

\subsection{Notation}

We use caligraphic uppercase letters to denote sets and set correspondences (e.g., $\calX$);
bold lowercase letters to denote vectors (e.g., $\price, \bm \pi$);
bold uppercase letters to denote matrices and vector-valued random variables (e.g., $\allocation$, $\bm \Gamma$);
lowercase letters to denote scalar quantities (e.g., $x, \gamma$);
and uppercase letters to denote scalar-valued random variables (e.g., $X, \Gamma$).
We denote the $i$th row vector of a matrix (e.g., $\allocation$) by the corresponding bold lowercase letter with subscript $i$ (e.g., $\allocation[\buyer])$. 
Similarly, we denote the $j$th entry of a vector (e.g., $\price$ or $\allocation[\buyer]$) by the corresponding lowercase letter with subscript $j$ (e.g., $\price[\good]$ or $\allocation[\buyer][\good]$).
Lowercase letters also denote functions: e.g., $f$ if the function is scalar valued, and $\f$ if the function is vector valued.
We denote the vector of ones of size $\numbuyers$ by $\ones[\numbuyers]$, the set of integers $\left\{1, \hdots, n\right\}$ by $[n]$, the set of natural numbers by $\N$, the set of real numbers by $\R$, and the postive and strictly positive elements of a set by a $+$ and $++$ subscript, respectively, e.g., $\R_+$ and $\R_{++}$. 
Finally, we denote the orthogonal projection operator onto a set $C$ by $\project[C]$, i.e., $\project[C](\x) = \argmin_{\y \in C} \left\|\x - \y \right\|^2$.

\subsection{Pseudo-games}

A (concave) \mydef{pseudo-game} \cite{arrow1954existence} $\pgame \doteq (\numplayers, \actionspace,  \actions, \actionconstr, \utilp)$ 
comprises $\numplayers \in \N_+$ players, each $\player \in \players$ of whom chooses an action $\action[\player] \in \actionspace[\player] \subset \R^{\numactions}$, with the players' joint action space $\actionspace = \bigtimes_{\player \in \players} \actionspace[\player]$.
Each player $\player$ aims to maximize their continuous utility $\utilp[\player]: \actionspace \to \R$, which is concave in $\action[\player]$, by choosing a feasible action from a set of actions $\actions[\player](\naction[\player]) \subseteq \actionspace[\player]$ determined by the actions $\naction[\player] \in \actionspace[-\player] \subset \R^\numactions$ of the other players, where $\actions[\player]: \actionspace[-\player] \rightrightarrows  \actionspace[\player]$ is a non-empty, continuous, compact- and convex-valued action correspondence.
For convenience, we represent each such correspondence as a set $\actions[\player](\naction[\player]) = \{ \action[\player] \in \actionspace[\player] \mid \actionconstr[\player][\numconstr](\action[\player], \naction[\player]) \geq \zeros, \text{ for all } \numconstr \in [\numconstrs]\}$, where for all $\numconstr \in [\numconstrs]$, $\actionconstr[\player][\numconstr]$ is a continuous and concave function in $\action[\player]$, which defines the constraints.
For notational convenience, we also define the \mydef{joint constraint function} $\constr = \left(\actionconstr[1], \hdots, \actionconstr[\numplayers] \right): \actionspace \to \R^{\numplayers \times\numconstrs}$.
If $\constr[\numconstr](\action)$ is concave in $\action$, for all $\numconstr \in [\numconstrs]$, then we say that the pseudo-game has \mydef{jointly convex constraints}, in which case the joint action correspondence is simply a convex set, i.e., $\actions = \{\action  \in \actionspace \mid \constr(\action) \geq \zeros\}$.
A pseudo-game is called \mydef{monotone}%
\footnote{We call a pseudo-game monotone if $- (\grad[{\action[1]}] \util[1], \hdots, \grad[{\action[\numplayers]}] \util[\numplayers])$ is a monotone operator. Such pseudo-games are also sometimes called dissipative pseudo-games, since $(\grad[{\action[1]}] \util[1], \hdots, \grad[{\action[\numplayers]}] \util[\numplayers])$ is called a dissipative operator if $- (\grad[{\action[1]}] \util[1], \hdots, \grad[{\action[\numplayers]}] \util[\numplayers])$ is a monotone operator.} 
if for all $\x, \y \in \actionspace$, $\sum_{\player \in \players} \left(\grad[{\action[\player]}] \utilp[\player](\x) - \grad[{\action[\player]}] \utilp[\player](\y) \right)^T \left( \x_\player - \y_\player \right) \leq 0$.
Finally, a (concave) \mydef{game} \cite{nash1950existence} is a pseudo-game where, for all players $\player \in \players$, $\actions[\player]$ is a constant correspondence, i.e., for all players $\player \in \players,$ $\actions[\player](\naction[\player]) = \actions[\player](\otheraction[-\player])$, for all $\action, \otheraction \in \actionspace$.

Given a pseudo-game $\pgame$, an $\varepsilon$-\mydef{generalized Nash equilibrium (GNE)} is a strategy profile $\action^* \in \actions(\action^*)$ s.t.\ for all $\player \in \players$ and $\action[\player] \in \actions[\player](\naction[\player][][][*])$, $\utilp[\player](\action^*) \geq \utilp[\player](\action[\player], \naction[\player][][][*]) - \varepsilon$.
An $\varepsilon$-\mydef{variational equilibrium (VE)} (or \mydef{$\epsilon$-normalized GNE}) of a pseudo-game \emph{with joint constraints\/} is a strategy profile $\action^* \in \actions$ s.t.\ for all $\player \in \players$ and $\action \in \actions$, $\utilp[\player](\action^*) \geq  \utilp[\player](\action[\player], \naction[\player][][][*]) - \varepsilon$.
A GNE (VE) is an $\varepsilon$-GNE (VE) with $\varepsilon = 0$.
While GNE are guaranteed to exist in all pseudo-games under standard assumptions (see \Cref{thm:existence_GNE}), VE are only guaranteed to exist in pseudo-games with jointly convex constraints (see \Cref{thm:jointly_convex_ve_gne}) \cite{arrow1954existence}.

\subsection{Fisher Markets with Social Influence}

In this paper, we study a model of Fisher markets with social influence, in which a buyer's utility may be influenced by the goods allocated to her neighbors.
A \mydef{Fisher market with social influence}, or an \mydef{influence Fisher market} for short, comprises $\numbuyers \in \N_+$ buyers and $\numgoods \in \N_+$ divisible goods.
Without loss of generality,
we assume that exactly one unit of each good $\good \in \goods$ is available.

The buyers are connected through a directed social influence graph $\graph=(\vertices, \edges)$, where $\vertices = \buyers$ is the set of buyers, and for any $\buyer, \buyerp\in \buyers$, there is an edge from $\buyerp$ to $\buyer$ iff the utility of $\buyer$ is influenced by the allocation $\allocation[\buyerp]$ of $\buyerp$.
We let $ \neighborset[\buyer]=\{\buyerp \mid (\buyerp, \buyer)\in \edges\}$ be the (incoming, and hence influential) neighbors of a buyer $\buyer$, and we define $\neighbordegree[\buyer]=|\neighborset[\buyer]|$, for all $\buyer\in \buyers$.

Each buyer $\buyer\in \buyers$ has a budget $\budget[\buyer] \in \Rp$ and a utility function $\util[\buyer]: \Rp^{(\neighbordegree[\buyer]+1) \times \numgoods} \to \R$ that depends on not only her own allocation, but also her neighbors'.
An instance of an influence Fisher market is thus given by a tuple $(\numbuyers, \numgoods, \graph, \util, \budget)$, where $\graph$ is the social network, $\util= \{\util[1],\hdots, \util[\numbuyers]\}$ is a set of utility functions, one per buyer, and $\budget\in \Rp^{\numbuyers}$ is a vector of buyer budgets.
When $\numbuyers$ and $\numgoods$ are clear from context, we denote influence Fisher markets simply by $(\graph, \util, \budget)$.

Given an influence Fisher market $(\graph, \util, \budget)$, an \mydef{allocation} $\allocation = (\allocation[1], \hdots, \allocation[\numbuyers])^{T}\in \Rp^{\numbuyers\times \numgoods}$ is a map from goods to buyers, represented as a matrix, s.t. $\allocation[\buyer][\good]\geq 0$ denotes the amount of good $\good\in \goods$ allocated to buyer $\buyer\in \buyers$. 
Likewise, we denote by $\allocation[{\neighborset[\buyer]}]=(\allocation[\buyerp])_{\buyerp\in \neighborset[\buyer]}^T \in \Rp^{\neighbordegree[\buyer]\times \numgoods}$ the matrix representing the bundles of goods obtained by buyer $\buyer$'s neighbors. 
A utility function is \mydef{locally non-satiated} if for all $\allocation[\buyer]\in \Rp^{\numgoods}, \allocation[\nei]\in \Rp^{\neighbordegree[\buyer]\times \numgoods}$, and $\varepsilon > 0$, there exists an $\allocationp[\buyer]\in \Rp^{\numgoods}$ with $\|\allocationp[\buyer]-\allocation[\buyer]\|\leq \varepsilon$ such that $\util[\buyer](\allocationp[\buyer], \allocation[\nei]) > \util[\buyer](\allocation[\buyer], \allocation[\nei])$.
Related, a utility function satisfies \mydef{no saturation} if $\forall \allocation[\buyer]\in \Rp^{\numgoods}$ and $\allocation[\nei]\in \Rp^{\neighbordegree[\buyer]\times \numgoods}$, there exists an $\allocationp[\buyer] \in \Rp^{\numgoods}$ such that $\util[\buyer](\allocationp[\buyer], \allocation[\nei]) > \util[\buyer](\allocation[\buyer], \allocation[\nei])$.
Note that if $\util[\buyer]$ is quasi-concave 
in $\allocation[\buyer]$ and satisfies no saturation, then it is locally non-satiated \citep{arrow1954existence}.

\mydef{Feasibility} asserts that
no more of each good $j$ is allocated than its available supply, i.e., $\forall \good\in \goods, \sum_{\buyer\in \buyers} \allocationstar[\buyer][\good] \leq 1$.
\mydef{Walras' law} states that, for each good $j$, either all of it is allocated, in which case its price is strictly positive, and otherwise, its price is zero.
Mathematically, $\sum_{\good\in \goods} \pricestar[\good](\sum_{\buyer\in \buyers} \allocationstar[\buyer][\good]-1) = 0$.

A tuple $(\allocationstar, \pricestar)$, which consists of an allocation $\allocationstar$ and prices $\pricestar = (\pricestar[1] \hdots, \pricestar[\numgoods])^{T} \in \Rp^{\numgoods}$, is a \mydef{competitive equilibrium (CE)} in an influence Fisher market $(\graph, \util, \budget)$ if (1)  fixing other buyers' allocations, buyers maximize their utilities constrained by their budget, i.e., $\forall \buyer\in \buyers, \allocationstar[\buyer] \in  \argmax_{\allocation[\buyer] \in \Rp^{\numgoods}: \allocation[\buyer] \cdot \pricestar \leq \budget[\buyer]} \util[\buyer](\allocation[\buyer],\allocationstar[\nei] )$, and (2) feasibility and Walras' law hold.

A Fisher market is a special influence Fisher market $(\graph, \util, \budget)$ where $\graph=(\vertices, \edges)$ satisfies $\edges=\emptyset$. In other words, each buyer $\buyer$ is isolated, so her utility $\util[\buyer]:\Rp^{\numgoods}\to \Rp$ depends only on her own allocation. As $\graph$ is simply a graph with $\numbuyers$ vertices and no edges, we can denote a Fisher market by the tuple $(\util, \budget)$.

When $\util$ is a set of specific utility functions, we refer to the influence Fisher market $(\graph, \util, \budget)$ by the name of the utility function: e.g., if $\util$ is a set of linear utility functions then $(\graph, \util, \budget)$ is a linear influence Fisher market.


\amy{i want to say here that FMs are a special case of AD markets, but we have not defined AD markets. i need this, b/c we need to justify the constraint that ``the sum of the prices equals the sum of the budgets.'' Deni suggested that this constraint would be implied by the fact that budgets are expressed in terms of units a numeraire good. and that the value of the budget has to equal the price of the real goods, for each buyer, b/c o/w buyers would gravitate towards collecting more of the numeraire good than the other goods. iow, this constraint seems part of the def'n of FMs, and it requires explanation.} \sadie{Deni added something in the next section, but I don't know if that's enough. I somehow think define AD markets and reduce it to Fisher market is not trivial, and I'm afraid it may cause more confusions. }

\section{Existence of Competitive Equilibrium via Pseudo-Games}
\label{sec:existence}

\amy{IMPORTANT: the auctioneer in this section does not seem very fictitious to me. it is a real player. otoh, it is fictitious in the buyer-only pseudo-game, since it arises thru the dual. so i think it is confusing to call the auctioneer fictitious here. i think we should reserve this terminology for the next section/tatonnement.}\sadie{I agree with this! I will check we use auctioneer instead of fictitious auctioneer in this section!}

In this section, we investigate the properties of competitive equilibrium in Fisher markets with social influence.
Our main tool is the pseudo-game (or abstract economy) model introduced by \citeauthor{arrow1954existence}, as both a generalization of the standard normal-form game in game theory and of the Arrow-Debreu market in microeconomics \citep{arrow1954existence}.
We provide a proof of existence of competitive equilibrium in influence Fisher markets, using methods similar to those employed by \citeauthor{arrow1954existence} in their seminal proof of the existence of competitive equilibria in Arrow-Debreu economies.

Following Arrow and Debreu, we define a pseudo-game with an auctioneer who sets prices.
Our pseudo-game then both generalizes and specializes theirs.
While in theirs, each trader's utility depends only on their own allocations, ours captures social influence through augmented utility functions.
While in theirs, buyers are constrained by their endowment, in ours, buyers are constrained by budgets.

Specifically, we construct an auctioneer-buyer pseudo-game comprising a single auctioneer and $\numbuyers$ individual buyers in which the auctioneer sets the good prices, while the buyers choose their allocations. 
Given an allocation $\allocation \in \R^{\numbuyers \times \numgoods}$, let $\excessd = \left( \sum_{\buyer\in \buyers} \allocation[\buyer] \right) - \ones[\numgoods]$ be the vector of \mydef{excess demands}, i.e., the total amount by which the demand for each good exceeds its supply. 
In our auctioneer-buyer pseudo-game, each buyer $\player$ chooses allocations $\allocation[\player]$ that maximize her utility subject to her budget constraint, given prices $\price$ determined by the auctioneer, the auctioneer chooses prices that maximize her total profit, i.e., $\price \cdot \excessd$,
fixing the allocation $\allocation$, subject to Walras' law. 
More specifically, we assume a numeraire (i.e., a good whose price we normalize to 1), and we view the buyers' budgets as quantities of this numeraire, in which case Walras' law can be restated as the sum of the prices being equal to the sum of the budgets: i.e., $\sum_{\good \in \goods}\price[\good] = \sum_{\buyer \in \buyers}\budget[\buyer]$.

In what follows, we show that the set of GNE in this auctioneer-buyer pseudo-game corresponds to the set of CE in an influence Fisher market; existence of a CE thus follows from existence of GNE.

\amy{discuss!} \sdeni{}{Importantly, in this pseudo-game we require the auctioneer to choose prices whose sum is equal to the sum of the budgets, because the budgets correspond are a numeraire good, i.e., their price is set to 1 (this is without loss of generality since in any Arrow-Debreu market only the proportions of the prices to one another matters), hence the prices of every other good have to be scaled appropriately, i.e., by the sum of the budget, so as to preserve the internal consistency of the Arrow-Debreu price system that emerges at the competitive equilibrium prices (since Arrow-Debreu assume wlog that prices are in the unit simplex).} \deni{Alternatively, we can get rid of this action space restriction by adding an additional seller player who has all the the goods in the market, considering budgets to be part of the commodity space and then restricting price to be in the unit simplex just like Arrow-Debreu does, I highly discourage this, it is so incomprehensible to a novice reader imo.}

\begin{assumption}
\label{assumption:existence_assum}
The influence Fisher market $(\graph, \util, \budget)$ satisfies for each buyer $\buyer \in \buyers$, $\util[\buyer]$ is
1.~continuous in $(\allocation[\buyer], \allocation[\nei])$,
2.~concave in $\allocation[\buyer]$, and
3.~satisfies no saturation.
\end{assumption}

\begin{definition}[Auctioneer-Buyer Pseudo-game] 
\label{def:Auctioneer-Buyer_pseudo}
Let $(\graph, \util, \budget)$ be an influence Fisher market.
The corresponding \mydef{auctioneer-buyer pseudo-game} $\pgame=(\numbuyers+1, \actionspace, \actions, \constr, \utilp)$ is defined by
\begin{itemize}
    \item an auctioneer and $\numbuyers$ buyers.
    
    \item Each buyer chooses an allocation $\allocation[\buyer] \in \actionspace[\buyer] = \Rp^{\numgoods}$, while the auctioneer chooses prices $\price \in \actionspace[\priceplayer] = \Rp^{\numgoods}$. 
    
    \item For all buyers $\buyer \in \buyers$,
    the feasible action set given the actions of other players is $\actions[\buyer] (\allocation[- \buyer], \price) = \{ \allocation[\buyer] \in \actionspace[\buyer] \mid \constr( \allocation[\buyer], \allocation[- \buyer], \price ) = \budget[\buyer]- \allocation[\buyer] \cdot \price \geq \zeros\}$.
    
    \item For the auctioneer, the feasible action set is the fixed set $\actions[\priceplayer] = \{ \price \in \actionspace[\priceplayer] \mid \price \cdot \ones[\numgoods] = \budget \cdot \ones[\numbuyers]\}$.
    
    \item For all players $\player \in [\numbuyers+1]$, $\player$ maximizes her utility $\utilp[\buyer]: \bigtimes_{\player \in [\numbuyers+1]} \actionspace[\player] \to \R$,
    defined by $\utilp[\buyer] (\allocation, \price) = \util[\buyer] (\allocation[\buyer], \allocation[\nei])$, for the buyers $\buyer \in [\numbuyers]$, and $\utilp[\priceplayer] (\allocation, \price) = \price\cdot \excessd$ 
    for the auctioneer.
\end{itemize}
\end{definition}

Next, we prove the existence of CE in influence Fisher markets that satisfy \Cref{assumption:existence_assum}.%
\footnote{Proofs of all theorems appear in the appendix.}

\begin{restatable}{theorem}{thmExistence}
\label{thm:competitive_equ_existence}
The set of competitive equilibria of any influence Fisher market $(\graph, \util, \budget)$ that satisfies \Cref{assumption:existence_assum} is equal to the set of generalized Nash equilibria of the associated auctioneer-buyer pseudo-game $\pgame=(\numbuyers+1, \actionspace, \actions, \constr, \utilp)$.
\end{restatable}

Existence of a CE in an influence Fisher market now follows immediately from existence of GNE in pseudo-games:

\begin{restatable}{corollary}{thmExistence}
\label{cor:competitive_equ_existence}
There exists a CE $(\allocationstar, \pricestar)$ in all influence Fisher markets $(\graph, \util, \budget)$ satisfying \Cref{assumption:existence_assum}.
\end{restatable}

\section{Computation of Competitive Equilibrium via Pseudo-Games}

Although we have established the existence of competitive equilibrium in all influence Fisher markets with continuous and concave utility functions, the proof itself provides little insight into equilibrium computation, as computing a GNE is PPAD-hard in general \citep{daskalakis2009complexity}. 
In order to gain further computational insights, we focus on a subset of influence Fisher markets in which each buyer's utility function is also homogeneous in its own allocation.

A utility function $\util[\buyer]$ is \mydef{homogeneous} in $\allocation[\buyer] \in \Rp^{\numgoods}$ if $\util[\buyer](\allocation[\buyer], \allocation[\nei])$ is homogeneous for all $\allocation[\nei] \in \Rp^{\neighbordegree[\buyer] \times \numgoods}$, i.e., $\util[\buyer] (\lambda \allocation[\buyer], \allocation[\nei]) = \lambda \util[\buyer] (\allocation[\buyer], \allocation[\nei])$, for all $\lambda \geq 0$.
As above, we also assume continuity and concavity.
We call utility functions that satisfy all three of these assumptions CCH utility functions, and Fisher markets inhabited by buyers with such utility functions CCH Fisher markets.

We can compute competitive equilibria in CCH Fisher markets (without social influence) via the Eisenberg-Gale convex program and its dual \citep{eisenberg1959consensus}. 
To generalize this convex program to CCH influence Fisher markets, we propose another pseudo-game, which we call the buyer (only) pseudo-game, that is jointly convex, and whose variational equilibria correspond to CE allocations.
Moreover, we observe that the ``dual'' of this pseudo-game
simultaneously characterizes CE prices.
In other words, while the auctioneer-buyer pseudo-game explicitly models an auctioneer who updates prices in response to the buyers' behavior,
in the buyer (only) pseudo-game, the auctioneer is ``fictitious,'' as it is implicit in the dual.

If $(\util, \budget)$ is a CCH Fisher market, then an optimal solution $\allocationstar$ to the Eisenberg-Gale program (Eq.~\ref{eq:eisenberg_gale_primal}) constitutes a CE allocation, and an optimal solution to the Lagrangian that represents the feasibility constraints (Eq.~\ref{eq:eisenberg_gale_primal_feas_constr}) are the corresponding equilibrium prices \citep{devanur2002market, jain2005market}.

\paragraph{Primal:}
\begin{subequations}
\label{eq:eisenberg_gale_primal}
    \begin{align}
        &\max_{\allocation\in \Rp^{\numbuyers\times \numgoods}} 
        &\sum_{\buyer \in \buyers} \budget[\buyer] \log(\util[\buyer] (\allocation[\buyer])) \tag{\ref{eq:eisenberg_gale_primal}}\\
        &\text{subject to}
        & \forall \good\in \goods, \:\sum_{\buyer \in \buyers}\allocation[\buyer][\good] \leq 1 \label{eq:eisenberg_gale_primal_feas_constr}
    \end{align}
\end{subequations}

In Fisher markets (without social influence), each buyer's utility maximization problem is independent of the others', as each depends only on the buyer's own allocation.
The Eisenberg-Gale program takes advantage of this independence.
It takes an aggregate perspective, maximizing the \emph{sum\/} of the buyers' utilities subject to their feasibility constraints, and nonetheless computes an optimal allocation that maximizes each buyer's \emph{individual\/} utility.

In influence Fisher markets, however, where this independence assumption does not hold, we can no longer compute CE from this aggregate perspective.
In our solution---CE as the VE of a jointly-convex pseudo-game---each buyer maximizes her own utility, subject to a shared feasibility constraint.
Note that the only players in this buyer pseudo-game are the $\numbuyers$ buyers; there is no auctioneer updating prices based on the buyers' behavior.

\begin{definition}[Buyer Pseudo-game] \label{def:Buyer_pseudo}
Let $(\graph, \util, \budget)$ be an influence Fisher market.
The corresponding jointly-convex \mydef{buyer pseudo-game} $\pgame=(\numbuyers, \actionspace, \actions, \constr, \utilp)$ is defined by 
\begin{itemize}
    \item For all buyers $\buyer \in \buyers, \actionspace[\buyer] = \Rp^{\numgoods}$.
    
    \item For all buyers $\buyer \in \buyers$, the feasible action set given the actions of other players is $\actions[\buyer] (\allocation[- \buyer]) = \{ \allocation[\buyer] \in \actionspace[\buyer] \mid \constr( \allocation[\buyer], \allocation[- \buyer] ) = \ones - \sum_{\buyer \in \buyers} \allocation[\buyer] \geq \zeros\}$.
    
    \item For all buyers $\buyer \in \buyers$, $\buyer$ maximizes her utility $\utilp[\buyer]: \bigtimes_{\player \in \players} \actionspace[\player] \to \R$ defined by $\utilp[\buyer] (\allocation) = \budget[\buyer] \log (\util[\buyer] (\allocation[\buyer], \allocation[\nei]))$.
\end{itemize}
\end{definition}

\begin{assumption}
\label{assumption:pseudo_assum}
For each buyer $\buyer \in \buyers$, $\util[\buyer]$ is
1.~continuous in $(\allocation[\buyer], \allocation[\nei])$; and \sadie{add remark in the appendix relating to cont. diff.}\sadie{Added.}\amy{Deni, please take a look.}
2.~concave and homogeneous%
\footnote{We note that homogeneity implies no saturation, since for all $\x \in \R_+^\numgoods$ and $\allocation[\nei] \in \Rp^{\neighbordegree[\buyer]\times \numgoods}$, there exists an allocation $(1+\varepsilon)\x$ for some $\epsilon > 0$ s.t.\@ $\util[\buyer]((1+\varepsilon)\x, \allocation[\nei]) = (1+\varepsilon)\util[\buyer](\x, \allocation[\nei])> \util[\buyer](\x, \allocation[\nei])$.}
in $\allocation[\buyer]$.
%
\end{assumption}

\begin{restatable}{theorem}{thmPseudoPrimal}
\label{thm:pseudo_game_equ}
Let $(\graph, \util, \budget)$ be an influence Fisher market satisfying \Cref{assumption:pseudo_assum}. 
Then, $\allocationstar$ is a CE allocation of $(\graph, \util, \budget)$ if and only if it is a variational equilibrium (VE) of the corresponding buyer pseudo-game $\pgame=(\numbuyers, \actionspace, \actions, \constr, \utilp)$.
Moreover, if $\allocationstar$ is a VE of the buyer pseudo-game, then the corresponding KKT conditions are satisfied with optimal Langrange multipliers $\pricelangstar[1] = \hdots = \pricelangstar[\numbuyers] = \pricestar$, which correspond to CE prices.

\end{restatable}

The construction of competitive equilibrium via the auctioneer-buyer pseudo-game (\Cref{thm:competitive_equ_existence}) is more general than the construction of competitive equilibrium via the buyer pseudo-game (\Cref{thm:pseudo_game_equ}); however, the existence of the auctioneer precludes monotonicity, and hence polynomial-time computability.
To obtain efficient algorithms, we assume the buyers' utilities are concave not only in their own allocations but in one another's allocations as well, which implies monotonicity.
We also require twice-continuous differentiability.

\begin{assumption}
\label{assumption:comp_pseudo_assum}
For each buyer $\buyer \in \buyers$,
1.~The conditions in \Cref{assumption:pseudo_assum}, and
2.~$\util[\buyer]$ is jointly concave: i.e., concave in $(\allocation[\buyer], \allocation[\nei])$, and
3.~and twice-continuously differentiable in $(\allocation[\buyer], \allocation[\nei])$. 
\end{assumption}

Under \Cref{assumption:comp_pseudo_assum}, an influence Fisher market  can be expressed as a monotone variational inequality. 
There exist methods
that converge in last iterates%
\footnote{\citet{Solodov1999ExtraProximal} and \citet{Ryu2019ODEAO} also provide methods that guarantee average-iterate convergence with this same rate in monotone variational inequalities.} to a solution of any monotone variational inequality at a rate of $O(\nicefrac{1}{T})$  (e.g., \citet{gorbunov2022extragradient}).
{Our next theorem follows from these two assertions:

\begin{restatable}{theorem}{thmPseudoConvergence}
There exist methods that converge in last iterates to the CE allocations of influence Fisher markets at a rate of $O(\nicefrac{1}{T})$ under \Cref{assumption:comp_pseudo_assum}. 
In such markets, approximate competitive equilibrium allocations can be computed in polynomial time.
\end{restatable}

\deni{Here are two papers to cite ``Solodov, M. V. and Svaiter, B. F. (1999). A hybrid approximate extragradient–proximal point algorithm
using the enlargement of a maximal monotone op-
erator.'' and ``Ryu, E. K., Yuan, K., and Yin, W. (2019). Ode analysis of stochastic gradient methods with optimism
and anchoring for minimax problems.''} \sadie{I though we were using results from "Extragradient Method: O(1/K) Last-Iterate Convergence for Monotone Variational Inequalities and Connections With Cocoercivity" by Eduard Gorbunov?} \deni{Either of the three is okay. The papers that are not from Gorbunov are more standard but they give average iterate convergence, the paper from Gorbunov gives last iterate. I think you can cite Gorbunov and use extragradient descent! Maybe add these other references as other possible algorithms in a footnote?}

\section{Computation of Competitive Equilibrium via Stackelberg Games}\label{sec:stackelberg}

Recently, \citet{cole2019balancing} presented a generalization of the Eisenberg-Gale dual for arbitrary CCH utility functions, which accurately characterizes competitive equilibrium prices, but fails to match the optimal objective value of the Eisenberg-Gale primal.
Building on their results, \citet{Goktas2021aConsumertheoretic} derived the exact Eisenberg-Gale dual, for which strong duality holds.

\paragraph{Dual:}
\begin{subequations}\label{eq:eisenberg_gale_dual}
    \begin{align}
        &\min_{\price\in \Rp^{\numgoods}}
        &&\sum_{\good\in \goods}\price[\good]
        + \sum_{\buyer \in \buyers} \budget[\buyer]
        \log \left( \util[\buyer] (\allocationstar[\buyer]) \right) - \budget[\buyer]
        \tag{\ref{eq:eisenberg_gale_dual}}\\
        &\text{s.t.} \:
        &&\forall \buyer \in \buyers, \; \allocationstar[\buyer] \in  \argmax_{\allocation[\buyer] \in \Rp^{\numgoods}: \allocation[\buyer] \cdot \price\leq \budget[\buyer]} \util[\buyer] (\allocation[\buyer])
    \end{align}
\end{subequations}

\noindent
We begin this section by deriving the ``duals'' of our buyer pseudo-game.


In the buyer pseudo-game $\pgame$, each buyer is solving an optimization problem (\Cref{eq:individual_primal}) in which they maximize their utility function by choosing an optimal action in their feasible action set, given the other buyers' VE actions.
Based on this observation, we can derive the ``dual'' of our buyer pseudo-game;
but as our pseudo-game comprises $\numbuyers$ different optimization problems, one for each buyer $\buyer \in \buyers$, instead of just one dual, we have $\numbuyers$ duals.
Moreover, because any VE of a jointly-convex pseudo-game satisfies the corresponding KKT conditions with optimal Langrange multipliers $\pricelangstar[1] = \hdots = \pricelangstar[\numbuyers]$ (\Cref{thm:jointly_convex_kkt} \cite{facchinei2009generalized}), all $\numplayers$ duals yield the same CE prices!
In other words, just as the dual of Eisenberg-Gale program characterizes the CE prices of a Fisher market, the $\numplayers$ duals of our pseudo-game characterize the CE prices of an influence Fisher market (satisfying \Cref{assumption:pseudo_assum}).

\begin{restatable}{theorem}{thmPseudoDual} 
\label{thm:pseudo_dual}
Let $(\graph, \util, \budget)$ be an influence Fisher market satisfying \Cref{assumption:pseudo_assum}, and let $\pgame$ be the corresponding buyer pseudo-game $\pgame=(\numbuyers, \actionspace, \actions, \constr, \utilp)$. 
For each buyer $\buyer \in \buyers$, fixing its neighbors' allocations $\allocationstar[\nei]$,
the dual of $\buyer$'s optimization problem, 
\begin{align} 
\label{eq:individual_primal}
    \max_{\allocation[\buyer] \in \Rp^{\numgoods}:
    \allocation[\buyer] + \sum_{k\neq \buyer}
    \allocationstar[k] \leq \ones}
    \budget[\buyer] \log (\util[\buyer] (\allocation[\buyer], \allocationstar[\nei]))
\end{align}
is given by 

\begin{subequations}
\label{eq:individual_dual}
    \begin{align}
        &\min_{\price\in \Rp^{\numgoods}}
        &&\sum_{\good\in \goods} \price[\good] \left( 1 - \sum_{k \neq \buyer} \allocationstar[k][\good] \right) + \budget[\buyer] \log( \util[\buyer] (\allocationstar[\buyer], \allocationstar[\nei])) - \budget[\buyer]
        \tag{\ref{eq:individual_dual}}\\
        &\text{s.t.}
        && \allocationstar[\buyer] \in \argmax_{\allocation[\buyer] \in \Rp^{\numgoods}: \allocation[\buyer] \cdot \price \leq \budget[\buyer]} \util[\buyer] (\allocation[\buyer], \allocationstar[\nei])
    \end{align}
\end{subequations}
\end{restatable}

\citet{goktas2021minmax} further show that the dual of the Eisenberg-Gale program can be re-expressed as the solution to the following zero-sum convex-concave Stackelberg game characterizes the CE of any CCH Fisher market:
%
\begin{align}
\label{eq:fisher_stackelberg}
    \min_{\price \in \Rp^{\numgoods}}
    \max_{\allocation \in \Rp^{\numbuyers\times \numgoods}: \allocation \cdot \price \leq \budget}
    \sum_{\good \in \goods} \price[\good] 
    + \sum_{\buyer \in \buyers} \budget[\buyer] \log (\util[\buyer] (\allocation[\buyer]))
\end{align}

\noindent
The leader in this game is a fictitious auctioneer (i.e., price setter), while the follower represents a set of buyers who effectively play as a team.
The objective function is the sum of the auctioneer's welfare (i.e., the sum of the prices) and the buyers' Nash social welfare.
\citeauthor{goktas2021minmax} also derive a first-order method that solves this game, which, via the aforementioned interpretation, can be understood as computing a competitive equilibrium of a Fisher market via \emph{t\^atonnement}.

We argue that competitive equilibria in \emph{influence\/} Fisher markets can likewise be characterized via Stackelberg equilibria.
This more general setting requires not just one, but a system of $\numbuyers$ 
zero-sum convex-concave Stackelberg games \cite{goktas2021minmax}, one per buyer. 
In each game, the leader once again is a fictitious auctioneer (i.e., price setter), but the follower is just an individual buyer, not the set of all buyers.
Moreover, in each buyer's Stackelberg game, the objective function is the sum of the auctioneer's revenue (i.e., the sum of the good prices, each one discounted by the supply available to buyer $\buyer$ beyond what has been claimed by the others) and the individual buyer's utility.
These $\numbuyers$ Stackelberg games are played simultaneously
with the fictitious auctioneer optimizing prices assuming all the buyers simultaneously best respond (i.e., play a Nash equilibrium), and the buyers best respond to the auctioneer's prices, given the other buyers' allocations.


\begin{definition}[Buyer $\buyer$'s Stackelberg Game]
Let $(\graph, \util, \budget)$ be an influence Fisher market.
The corresponding \mydef{Stackelberg game for buyer $\buyer$} is defined by
\begin{align}
    \min_{\price\in \Rp^{\numgoods}} 
    \max_{\allocation[\buyer] \in \Rp^{\numgoods}: \allocation[\buyer] \cdot \price \leq \budget[\buyer]}
    &\sum_{\good\in \goods} \price[\good] \left( 1 - \sum_{k \neq \buyer} \allocationstar[k][\good] \right)
    + \notag \\
    &\budget[\buyer] \log( \util[\buyer] (\allocation[\buyer], \allocationstar[\nei]))  \label{eq:buyer_stackelberg}
\end{align}
\end{definition}

The following corollary follows from Theorems~\ref{thm:pseudo_game_equ} and~\ref{thm:pseudo_dual}.

\begin{restatable}{corollary}{corollaryStackelberg}
\label{cor:stackelberg_iff_ce}
$(\allocationstar, \pricestar)$ is a competitive equilibrium of an influence CCH Fisher market $(\graph, \util, \budget)$ satisfying \Cref{assumption:pseudo_assum} iff $(\allocationstar[\buyer], \pricestar)$ is a Stackelberg equilibrium in buyer $\buyer$'s Stackelberg game, for all buyers $\buyer \in \buyers$.
\end{restatable}
\begin{proof}  
  For all $\buyer\in \buyers$, $(\allocationstar[\buyer], \pricestar)$ is a Stackelberg equilibrium in buyer $\buyer$'s Stackelberg game iff $(\allocationstar[\buyer], \pricestar)$ solves 
    \begin{align}
    &\min_{\price\in \Rp^{\numgoods}} \! \!
    &&\sum_{\good\in \goods} \price[\good] \! \! \left( \! 1 - \sum_{k \neq \buyer} \allocationstar[k][\good] \!  \right) \! \!
    + \budget[\buyer] \log( \util[\buyer] (\allocationstar[\buyer], \allocationstar[\nei])) - \budget[\buyer] \nonumber \\
    &\text{s.t.}\:
    &&\allocationstar[\buyer] \in \argmax_{\allocation[\buyer] \in \Rp^{\numgoods}: \allocation[\buyer] \cdot \price \leq \budget[\buyer]} \util[\buyer] (\allocation[\buyer], \allocationstar[\nei])
\end{align}

\noindent
By \Cref{thm:pseudo_dual}, $\pricestar$ is a solution to this bi-level optimization problem (\Cref{eq:individual_dual}) iff $\allocationstar[\buyer]$ is a solution to buyer $\buyer$'s optimization problem (\Cref{eq:individual_primal}) in the buyer pseudo-game corresponding to $(\graph, \util, \budget)$.
Finally, by \Cref{thm:pseudo_game_equ}, $(\allocationstar, \pricestar)$ is a competitive equilibrium of $(\graph, \util, \budget)$.
\end{proof}

Our Stackelberg game formulation of CE in influence Fisher markets enables us to compute CE by solving a system of buyer Stackelberg games: i.e., solving for a Stackelberg equilibrium in each of the buyer Stackelberg games together with a Nash equilibrium among the buyers in the system.
Towards that end, for convenience, we define the \mydef{objective function} for buyer $\buyer$'s Stackelberg game:
\begin{align} 
\label{eq:obj_func}
 \obj_{\buyer}(\allocation[\buyer], \price)
 &\coloneqq
 \sum_{\good\in \goods} \price[\good] \left( 1 - \sum_{k \neq \buyer} \allocationstar[k][\good] \right)
    + \budget[\buyer] \log( \util[\buyer] (\allocation[\buyer], \allocationstar[\nei]))
\end{align}

\noindent
and the $i$th (fictional) auctioneer's \mydef{value function} in buyer $\buyer$'s Stackelberg game:
\begin{align}
\label{eq:value_func}
    \val_{\buyer}(\price)
    &\coloneqq
    \max_{\allocation[\buyer] \in \Rp^{\numgoods}: \allocation[\buyer] \cdot \price \leq \budget[\buyer]}
    \sum_{\good\in \goods} \price[\good] \left( 1 - \sum_{k \neq \buyer} \allocationstar[k][\good] \right) \notag \\
    & \quad \quad \quad \quad \quad \quad \quad  + \budget[\buyer] \log( \util[\buyer] (\allocation[\buyer], \allocationstar[\nei]))
\end{align}

\noindent
Moreover, while each buyer is playing a Stackelberg game with its fictitious auctioneer, all buyers are also playing an $\numbuyers$-buyer simultaneous game with one another, in which each buyer maximizes its objective function $\obj_{\buyer}(\allocation[\buyer], \price)$ (\Cref{eq:obj_func}), given the prices $\price$ set by the auctioneer and the other buyers' allocations.
We can characterize a Nash equilibrium of this $\numbuyers$-buyer game as follows:
\begin{subequations}
\label{eq:normal_form_game}
\begin{align}
\allocationstar[\buyer] 
&\in \argmax_{\allocation[\buyer] \in  \Rp^{\numgoods}: \allocation[\buyer] \cdot \price \leq \budget[\buyer]} \sum_{\good\in \goods} \price[\good] \left( 1 - \sum_{k \neq \buyer} \allocationstar[k][\good] \right) \notag \\
    & \quad \quad \quad \quad \quad \quad \quad + \budget[\buyer] \log( \util[\buyer] (\allocation[\buyer], \allocationstar[\nei]))
    \tag{\ref{eq:normal_form_game}} \\
    &= \argmax_{\allocation[\buyer] \in \Rp^{\numgoods}: \allocation[\buyer] \cdot \price \leq \budget[\buyer]} \util[\buyer] (\allocation[\buyer], \allocationstar[\nei]) 
\end{align}
\end{subequations}
\noindent
As the first summand in Equation~\ref{eq:normal_form_game} and $\budget[\buyer]$ are constants (i.e., they do not depend on $\allocation[\buyer]$), and $\log$ is a monotonic function,
buyer $i$ simply seeks to maximize its utility subject to its budget constraint.

Using a subdifferential envelope theorem \citep{goktas2021minmax}, we now derive the subgradient of each auctioneer's value function $\val_{\buyer}$ (\Cref{eq:value_func}).

\begin{restatable}{theorem}{thmSubdiff}
\label{thm:subdiff_equal_excess_demands}
Given an influence CCH Fisher market $(\graph, \util, \budget)$, the subdifferential of the $i$th auctioneer's value function in buyer $\buyer$'s Stackelberg game (\Cref{eq:value_func}) at given prices $\price$ is equal to the negative excess demand 
at $\price$: i.e.,
%
    $\subdiff[\price]
        \val_{\buyer}(\price)
        = \ones - \sum_{\buyer \in \buyers} \allocationstar[\buyer]$.
\end{restatable}

Interestingly, this subgradient turns out to equal the negative excess demand in the market at the given prices.
As excess demand is an aggregate quantity, it is independent of buyer $i$.
Indeed, the subgradients of \emph{all\/} the fictional auctioneers are the same; so there is effectively just one auctioneer.

Based on this observation, we now present our \mydef{Nash Equilibrium (NE)-oracle gradient descent} algorithm (\Cref{alg:ne_oracle_gd}), which follows the subgradient of the auctioneer's value function, assuming access to a NE-oracle.
Given prices $\price \in \Rp^{\numgoods}$, this oracle returns a Nash equilibrium $\allocationstar$ of the $\numbuyers$-buyer concave game specified by \Cref{eq:normal_form_game}.
The algorithm then runs subgradient descent on the auctioneer's value function.
Overall, this approach corresponds to solving for a CE allocation and 
prices via \emph{t\^{a}tonnement}, assuming the NE oracle is exact.
As NE-oracles are rarely exact, \Cref{alg:ne_oracle_gd} assumes a NE-oracle that finds a Nash equilibrium up to some approximation error $\delta$.

Finally, under standard assumptions (i.e., \Cref{assumption:comp_pseudo_assum}), the auctioneer's value function (\Cref{eq:value_func}) is convex and $\lipcont[\val]$-Lipschitz continuous in $\price$ with $\lipcont [\val] = \max_{\price\in \Rpp^{\numgoods}} \| \grad[\price] \val (\price) \|$.%
\footnote{Although $\grad[\price] \val (\price)$ is not necessarily bounded at $\price=0$, we can remedy this fact by shifting $\price$ by a small constant $\varepsilon > 0$, albeit losing some accuracy.}
These properties imply that our NE-oracle gradient descent algorithm converges to 
competitive equilibrium at a rate of $O(\nicefrac{1}{\sqrt{T}})$.
We include a more detailed statement 
of the following theorem in the appendix. 

\begin{restatable}{theorem}{thmStackelbergConvergence}
\label{thm:convergence_of_tatonnment}
\Cref{alg:ne_oracle_gd} (i.e., \emph{t\^atonnement}) converges to a competitive equilibrium in any influence CCH Fisher market $(\graph, \util, \budget)$ satisfying \Cref{assumption:comp_pseudo_assum} at a rate of $O(\nicefrac{1}{\sqrt{T}})$.
\end{restatable}

\begin{remark}
We can implement an approximate NE oracle by computing the buyers' equilibrium allocations via extragradient ascent~\cite{gorbunov2022extragradient}, which is guaranteed to converge to a Nash equilibrium at a rate of $O(\nicefrac{1}{T})$, as the $\numplayers$-buyer concave game defined by \Cref{eq:normal_form_game} is monotone.
This observation gives rise to \Cref{alg:nested_ne_gd} (see Appendix, \nameref{app:algo} Section), which computes a competitive equilibrium in influence CCH Fisher markets in polynomial time.
\end{remark}

\begin{algorithm}[htbp]
\caption{NE-Oracle \emph{T\^{a}tonnement\/} For Influence Fisher Markets}
\textbf{Inputs:} $\graph, \util, \budget, \price^{(0)}, \learnrate, \delta$\\
\textbf{Outputs:} $\allocationstar, \pricestar$
\label{alg:ne_oracle_gd}
\begin{algorithmic}[1]
\For{$\iter = 1, \hdots, \iters$}
    \State Find $\allocationp\in \Rp^{\numbuyers\times \numgoods}$ with $\allocationp\cdot\price^{(\iter-1)} \leq \budget$ such that:
    \State for all $\buyer \in \buyers$, $\util[\buyer] (\allocationp[\buyer], \allocationp[\nei])\geq \util[\buyer] (\allocation[\buyer], \allocationp[\nei]) - \delta$, \\
    \State for any $\allocation[\buyer] \in \Rp^{\numgoods}$ satisfying $\allocation[\buyer] \cdot\price^{(\iter-1)} \leq \budget[\buyer]$
    \State Set $\allocation^{(\iter)}=\allocationp$
    \State Set $\price^{(\iter)} = \project[\Rp^{\numgoods}] \left(\price^{(\iter-1)} - \learnrate(1 -\sum_{\buyer \in \buyers} \allocation[\buyer]^{(\iter)}) \right)$
\EndFor
\State \Return $\allocation^{(\iters)}, \price^{(\iters)}$
\end{algorithmic}
\end{algorithm}

\section{Experiments}

We ran a series of experiments%
\footnote{We include a detailed description of our experimental setup in the Appendix.}
to see how the empirical convergence rates of \Cref{alg:ne_oracle_gd} compare to its theoretical guarantees under various utility structures.
We considered three standard utility functions: linear, in which buyers practice utilitarian social welfare in their neighborhoods; Cobb-Douglas, in which practice Nash social welfare in their neighborhoods; and Leontief, in which practice egalitarian social welfare in their neighborhoods.
Each utility structure endows the objective function (\Cref{eq:obj_func}) and the value functions (\Cref{eq:value_func}) with different smoothness properties, which in turn varies the convergence properties of our algorithms. 

Let $\valuation[\buyer] \in \R^\numgoods$ be a vector of parameters that describes the utility function $\utilp[\buyer]: \Rp^{\numgoods}\to \Rp$ of buyer $\buyer \in \buyers$.
We consider the following (standard) utility functions: for all $\buyer \in \buyers$,
\begin{enumerate}
    \item Linear: $\util[\buyer](\allocation[\buyer],\allocation[\nei]) = \sum_{k\in {\buyer}\cup \nei} \utilp[k](\allocation[k])$, where $\utilp[\buyer](\allocation[\buyer])=\sum_{\good\in \goods} \valuation[\buyer][\good] \allocation[\buyer][\good]$
    
    \item Cobb-Douglas: $\util[\buyer](\allocation[\buyer], \allocation[\nei])=\prod_{k\in \{\buyer\}\cup \nei } \utilp[k](\allocation[k])$, where $\utilp[\buyer](\allocation[\buyer]) = \prod_{\good \in \goods} \allocation[\buyer][\good]^{\valuation[\buyer][\good]}$
    
    \item Leontief: $\util[\buyer](\allocation[\buyer], \allocation[\nei]) = \min_{k\in \{\buyer\}\cup \nei } \utilp[k](\allocation[k])$, where $\utilp[\buyer](\allocation[\buyer]) = \min_{\good \in \goods} \left\{ \frac{\allocation[\buyer][\good]}{\valuation[\buyer][\good]}\right\}$
\end{enumerate}

\begin{figure*}[htpb]
    \begin{subfigure}[b]{0.31\textwidth}
         \centering
         \includegraphics[width=\textwidth,height=5cm]{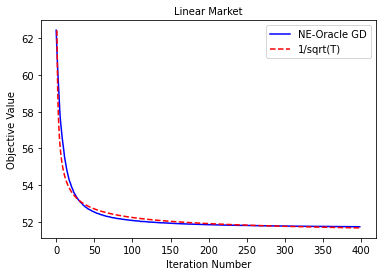}
         \label{fig:linear_market}
     \end{subfigure}
     \hfill
     \begin{subfigure}[b]{0.31\textwidth}
         \centering
         \includegraphics[width=\textwidth,height=5cm]{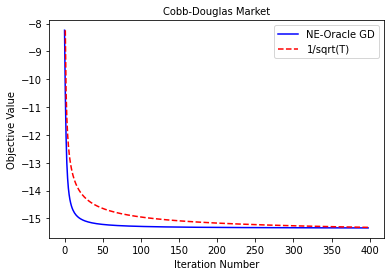}
         \label{fig:cd_market}
     \end{subfigure}
     \hfill
     \begin{subfigure}[b]{0.31\textwidth}
         \centering
         \includegraphics[width=\textwidth,height=5cm]{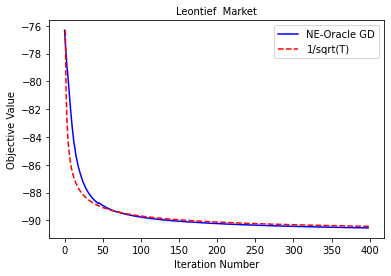}
         \label{fig:leontief_market}
     \end{subfigure}
    \caption{In \textcolor{blue}{blue}, we depict a trajectory of average value of the objective function across experiments (\Cref{eq:obj_func}), for \Cref{alg:ne_oracle_gd} with EG as the NE-oracle, in randomly initialized linear, Cobb-Douglas, and Leontief Fisher markets with social influence. In red, we plot an arbitrary $O(\nicefrac{1}{\sqrt{T}})$ function.}\label{fig:convergence_results}
\end{figure*}

Assuming any of these three utility functions, we can solve for a Nash equilibrium among buyers by formulating a monotone variational inequality problem, and solving it via the extragradient method (EG) in $O(\nicefrac{1}{T})$ iterations \cite{gorbunov2022extragradient}. 
Then, by using EG as the NE-oracle, we can efficiently compute an optimal $\allocationstar (\price)$ for any given $\price$, which yields \Cref{alg:nested_ne_gd} (see Appendix),
a specific implementation of \Cref{alg:ne_oracle_gd}.

\Cref{fig:convergence_results} depicts the empirical convergence of \Cref{alg:ne_oracle_gd} with EG as the NE-oracle.
We observe that convergence is fastest in influence Fisher markets with Cobb-Douglas utilities, followed by linear, and then Leontief. 
For influence Fisher markets with Cobb-Douglas utilities, both the value and the objective function are differentiable; in fact, they are both twice continuously differentiable, making them both Lipschitz-smooth.
These factors combined seem to lead to a faster convergence rate than $O(\nicefrac{1}{\sqrt{T}})$. 
On the other hand, for influence Fisher markets with linear utilities, we seem to obtain a tight convergence rate of $O(\nicefrac{1}{\sqrt{T}})$, which seems plausible, as the value function is not differentiable assuming linear utilities, and hence we are unlikely to achieve a better convergence rate. 
Finally, influence Fisher markets with Leontief utilities, in which the objective function is not differentiable, are the hardest markets for our algorithm to solve.
Nonetheless, we still observe a decent convergence rate, one that appears only slightly slower than $O(\nicefrac{1}{\sqrt{T}})$.

\section{Conclusion}

In this paper, we studied a special case of Arrow-Debreu markets with social influence, which we call Fisher markets with social influence, or influence Fisher markets for short.
First, we extended known results on the existence of competitive equilibrium in markets with social influence to a larger more natural class of markets.
Our proof proceeds by reducing an influence Fisher market to an auctioneer-buyer pseudo-game such that every generalized Nash equilibrium in the pseudo-game is a competitive equilibrium of the influence Fisher market.
The existence of generalized Nash equilibrium in pseudo-games thus implies the existence of competitive equilibrium in influence Fisher markets.

We then introduced a monotone jointly convex buyer-only pseudo-game as a generalization of the Eisenberg-Gale program, whose variational equilibria correspond to the competitive equilibria in influence Fisher markets.
In this pseudo-game, the duals of the individual buyers' utility-maximization problems constrained by the supply constraint comprise a system of $\numbuyers$ simultaneously-played zero-sum Stackelberg games, which simultaneously characterize the competitive equilibrium prices of the influence Fisher market.
We then show that running gradient descent on the leaders'/auctioneers' value functions in these games is equivalent to solving the market via a variant of \emph{t\^{a}tonnement}, where in addition to the auctioneers iteratively adjusting prices, the buyers iteratively learn a Nash equilibrium in response to these prices.

Our results pave the way for future work developing methods to compute competitive equilibria in more general types of influence markets beyond those considered in this paper \cite{Chen2011MakretwithSocialInfluence}, and other market models with graphical structure, such as graphical economies \cite{Kakade2004GraphicalE} \amy{insert reference to KKO paper}\sadie{Added}.

\section{Acknowledgments}
This research was partially supported by the National Science Foundation (CMMI-1761546).

\bibliographystyle{plainnat}  
\bibliography{references}  

\appendix

\newpage
\phantom{ }
\newpage

\section{Preliminaries} \label{app:prelim}
\begin{theorem}  \label{thm:existence_GNE} \cite{facchinei2009generalized}
Consider a pseudo-game $\pgame \doteq (\numplayers, \actionspace,  \actions, \actionconstr, \utilp)$ and suppose that for all players $\player \in \players$:
\begin{enumerate}
    \item $\actionspace[\player]$ is a nonempty, convex, and compact set.
    \item $\actions[\player]$ is continuous, i.e., upper and lower hemicontinuous, and for all action profiles $\naction[\player] \in \actionspace[-\player]$, $\actions[\player](\naction[\player])$ is nonempty, closed, and concave.
    \item $\utilp[\player](\cdot, \naction[\player])$ is quasi-concave on $\actions[\player](\naction[\player])$.
\end{enumerate}
Then a GNE exists. 
\end{theorem}

In this paper, we focus on pseudo-games with jointly convex constraints, for which a much more complete theory exists than general pseudo-games.
We first introduce the \mydef{variational inequality problem (VI)} $\vi[\actions][\objs]$, which consists of finding a vector $\actionstar\in \actions$ such that $(\otheraction-\actionstar)^T\objs(\actionstar)\geq 0$, for all $\otheraction\in \actions$. 
A valuable property of jointly convex pseudo-games is that we can reduce the problem of finding a GNE to solving a VI problem.

\begin{theorem}
\label{thm:jointly_convex_ve_gne}
\cite{facchinei2009generalized}
Let $\pgame \doteq (\numplayers, \actionspace,  \actions, \constr, \utilp)$ be a pseudo-game jointly convex such that for all players $\player \in \players$, the utility functions $\utilp[\player]$ are continuously differentiable. Let $\actions=\{\action\in \actionspace\mid \constr(\action)\geq \zeros\}$ be the joint action correspondence set, and let $\objs(\action)\coloneqq (\grad[{\action[\player]}]\utilp[\player](\action))_{\player\in \players}$. Then every solution of the variational inequality $\vi[\actions][\objs]$ is also a GNE of the pseudo-game $\pgame$. 
\end{theorem}
\ssadie{}{
\begin{remark}
Note that we can also relieve the constraints such that each utility functions $\utilp[\buyer]$ is continuous rather than continuously differentiable. When utility functions are continuous, the set of solutions of the generalized variational inequality $\vi[\actions][\objs]$ corresponds to the set of  GNE of the pseudo-game $\pgame$. 
\end{remark}
}

We call a GNE of a jointly convex pseudo-game that is also a solution to $\vi[\actions][\objs]$ a \mydef{variational equilibrium (VE)}.
Note that the set of VE is a subset of the set of GNE.
The converse, however, is not true, unless $\actionspace \subseteq \actions$.
Further, when $\pgame$ is a game, GNE and VE coincide; we refer to this set simply as NE.

\begin{theorem}\label{thm:jointly_convex_kkt}
\cite{facchinei2009generalized}
Let $\pgame \doteq (\numplayers, \actionspace,  \actions, \constr, \utilp)$ be a pseudo-game with jointly convex constraints such that $\utilp[\player]$, $\constr$ are continuously differentiable, then the following statement holds:
\begin{enumerate}
    \item Let $\actionstar$ be a solution of the $\vi[\actions][\objs]$ such that the KKT conditions hold with some multiplier $\pricelangstar$. Then $\actionstar$ is a GNE of the pseudo-game, and the corresponding KKT conditions are satisfies with $\pricelangstar[1]=\hdots=\pricelangstar[\numplayers]=\pricelangstar$.
    \item Conversely, assume that $\actionstar$ is a GNE of the pseudo-game $\pgame$ such that the KKT conditions are satisfies with $\pricelangstar[1]=\hdots=\pricelangstar[\numplayers]=\pricelangstar$. Then $(\actionstar, \pricelangstar)$ is a KKT point of $\vi[\actions][\objs]$, and $\actionstar$ itself is a solution of $\vi[\actions][\objs]$. 
\end{enumerate}
\end{theorem}

\section{Algorithms}\label{app:algo}

\begin{algorithm}[H]\label{alg:nested_ne_gd}
\caption{Nested-NE T\^{a}tonnement For Influence Fisher Markets}
\textbf{Inputs:} $\graph, \util, \budget, \learnrate[\allocation], \learnrate[\price], \allocation^{(0)}, \price^{(0)}$\\
\textbf{Outputs:} $\allocationstar, \pricestar$
\begin{algorithmic}[1]
\For{$\iterouter = 1, \hdots, \iters[\price]$}
    \For{$\iterinner=1, \hdots, \iters[\allocation]$}
        \State For all $\buyer\in \buyers$,
        $\allocation[\buyer]^{(\iterinner+1/2)} = \project[{\allocation[\buyer] \in \Rp^{\numgoods}: \allocation[\buyer] \cdot \price^{(\iterouter)} \leq \budget[\buyer]}]
        \left(\allocation[\buyer]^{(\iterinner)} + \learnrate[\allocation] \grad[{\allocation[\buyer]}]
        \util[\buyer] (\allocation[\buyer]^{(\iterinner)}, \allocation[\nei]^{(\iterinner)}) \right)$
        \State For all $\buyer\in \buyers$,
        $\allocation[\buyer]^{(\iterinner+1)} = \project[{\allocation[\buyer] \in \Rp^{\numgoods}: \allocation[\buyer] \cdot \price^{(\iterouter)} \leq \budget[\buyer]}]
        \left(\allocation[\buyer]^{(\iterinner)} + \learnrate[\allocation] \grad[{\allocation[\buyer]}]
        \util[\buyer] (\allocation[\buyer]^{(\iterinner+1/2)}, \allocation[\nei]^{(\iterinner+1/2)}) \right)$
    \EndFor
    \State Set $\allocation^{(\iter)}= \allocation^{(\iters[\allocation])}$
    \State Set $\price^{(\iterouter)} = \project[\Rp^{\numgoods}] \left(\price^{(\iterouter-1)} - \learnrate[\price] (1-\sum_{\buyer\in \buyers} \allocation[\buyer]^{(\iter)} \right)$
\EndFor
\State \Return $\allocation^{(\iters)}, \price^{(\iters)}$
\end{algorithmic}
\end{algorithm}

\section{Omitted Proof}\label{app:proof}

\thmExistence*
\begin{proof}
We will define a pseudo-game $\pgame$ whose generalized Nash equilibrium points will correspond to a competitive equilibrium in an influence Fisher market $(\graph, \util, \budget)$. For convenience, we denote that $\sum_{\buyer\in \buyers}\budget[\buyer]=\budget[ ]$.

First, there will be $\numbuyers+1$ players--the $\numbuyers$ buyers and a fictitious player who chooses prices, which is termed the price player. 
Each buyer will choose an allocation $\allocation[\buyer]\in \actionspace[\buyer]=\Rp^{\numgoods}$ and the price player will choose a price $\price\in \actionspace[\priceplayer]=\Rp^{\numgoods}$. 
For convenience, let $\allocation[-\buyer]$ denotes the allocations of other buyers except buyer $\buyer$, and let $\allocation$ denotes the allocations of all buyers. 
Then, for all buyer $\buyer\in \buyers$, the feasible action set given actions of other player is $\actions[\buyer](\allocation[-\buyer], \price) = \{ \allocation[\buyer] \in \actionspace[\buyer] \mid \actionconstr[\buyer](\allocation[\buyer], \allocation[-\buyer], \price)=\budget[\buyer]-\allocation[\buyer]\cdot \price \geq \zeros\}$, and for the price player, the feasible set is a fixed set $\actions[\priceplayer]=\{\price\in \actionspace[\priceplayer]\mid \ones^T\price= \budget[ ]\}$.
Finally, each buyer $\buyer$ is maximizing her utility defined by $\utilp[\buyer](\allocation,\price)=\util[\buyer](\allocation[\buyer], \allocation[\nei])$, while the price player maximizes her utility $\utilp[\priceplayer](\allocation, \price)=\price\cdot \excessd$ where $\excessd=\left(\sum_{\buyer\in \buyers}\allocation[\buyer]\right) - \ones$. 

By \Cref{thm:existence_GNE}, we know that there exists an GNE $(\allocationstar,\pricestar)$ in $\pgame$. Next, we want to show that the GNE $(\allocationstar, \pricestar)$ is a competitive equilibrium in influence Fisher market $(\graph, \util, \budget)$. 

First, by the definition of GNE, we know that $\forall\buyer\in \buyers$, $\utilp[\buyer](\allocationstar, \pricestar)\geq \utilp[\buyer](\allocation[\buyer], \allocationstar[-\buyer], \pricestar)$ for all $\allocation[\buyer]\in \actions[\buyer](\allocationstar[-\buyer], \pricestar)$. That is, for any $\allocation[\buyer]$ s.t. $\allocation[\buyer]\cdot\pricestar\leq\budget[\buyer]$, $\util[\buyer](\allocationstar[\buyer], \allocationstar[\nei])\geq \util[\buyer](\allocation[\buyer], \allocationstar[\nei])$ by our construction of $\pgame$. Thus, $(\allocationstar, \pricestar)$ is utility maximizing. 

Next, we will show that the market clears by showing that both Walras' law and the feasibility condition hold. 

To do so, we first show that $\forall \buyer\in \buyers$, $\allocationstar[\buyer]\cdot \pricestar=\budget[\buyer]$. 
For any $\buyer\in \buyers$, there exists $\allocation[\buyer]\in \R^{\numgoods}$ such that $\util[\buyer](\allocation[\buyer], \allocationstar[\nei])>\util[\buyer](\allocationstar[\buyer], \allocationstar[\nei])$, since $\util[\buyer](\cdot,\allocationstar[\nei])$ is quasi-concave, for any $0<t<1$, $\util[\buyer](t\allocationstar[\buyer]+(1-t)\allocation[\buyer], \allocationstar[\nei])>\util[\buyer](\allocationstar[\buyer], \allocationstar[\nei])$. If $\allocationstar[\buyer]\cdot\pricestar<\budget[\buyer]$, we can pick $t$ small enough such that $\allocationp[\buyer]=t\allocationstar[\buyer]+(1-t)\allocation[\buyer]$ satisfies $\allocationp[\buyer]\cdot\pricestar\leq \budget[\buyer]$ and $\util[\buyer](\allocationp[\buyer], \allocationstar[\nei])>\util[\buyer](\allocationstar[\buyer], \allocationstar[\nei])$, this contradicts the utility maximizing condition which we have proved. Then, summing across the buyers, we get $\sum_{\buyer\in \buyers}\allocationstar[\buyer]\cdot \pricestar=\budget[ ]
\implies \pricestar\cdot (\sum_{\buyer\in \buyers}\allocationstar[\buyer]) - \pricestar\cdot \ones=0$ $\implies \pricestar\cdot (\sum_{\buyer\in \buyers}\allocationstar[\buyer]-\ones)=\pricestar\cdot \excessdstar=0$. Therefore, $(\allocationstar, \pricestar)$ satisfies the Walras' law. 

Finally, let $e_{\good}\in \R^{\numgoods}$ be the vector in which every component is 0, except the $\good$th, which is 1. It is obvious that $e_{\good}\in \actions[\priceplayer]=\{\price\in \actionspace[\priceplayer]\mid \ones^T\price= 1\}$. Thus, since the price player's utility is maximized, we know that $\utilp[\priceplayer](\allocationstar, \pricestar)\geq \utilp[\priceplayer](\allocationstar, e_{\good})$ for all $\good\in \goods$. That is, $0=\pricestar\cdot \excessdstar \geq e_{\good}\cdot \excessdstar = \excessdstar[\good]\implies \forall \good\in \goods, \sum_{\buyer\in \buyers}\allocationstar[\buyer][\good]\leq 1$.
\end{proof}

\thmPseudoPrimal*
\begin{proof}
First, by \Cref{thm:jointly_convex_ve_gne}, there exists a VE $\allocationstar$ of $\pgame$, and we want to show that $\allocationstar$ is a competitive equilibirum allocation of the influence Fisher market  $(\graph, \util, \budget)$. 

For each buyer $\buyer\in \buyers$, the Lagrangian is given by:
\begin{align*}
    \lang[\buyer](\allocation[\buyer], \allocationstar[-\buyer], \pricelang[\buyer], \bmu[\buyer])
    &= -\budget[\buyer]\log(\util[\buyer](\allocation[\buyer], \allocationstar[\nei]))\\
    &+ \sum_{\good\in \goods} \pricelang[\buyer][\good] (\allocation[\buyer][\good] + \sum_{k\neq \buyer} \allocationstar[k][\good]-1) 
    + \sum_{\good\in \goods} \bmu[\buyer][\good] (-\allocation[\buyer][\good])
\end{align*}
and $(\allocationstar[\buyer], \pricelangstar[\buyer])$ satisfies the KKT conditions of buyer $\buyer$ iff
\begin{itemize}
    \item (Stationarity)$\grad[{\allocation[\buyer]}] \lang[\buyer](\allocationstar[\buyer], \allocationstar[-\buyer], \pricelangstar[\buyer], \bmustar[\buyer])=\zeros$.
    \item (Complementary Slackness) $\forall \good\in \goods,\; \pricelangstar[\buyer][\good](\sum_{\buyer\in \buyers} \allocationstar[\buyer][\good]-1)=0$, $\bmustar[\buyer][\good](-\allocationstar[\buyer][\good])=0$.
    \item (Primal Feasibility) $\forall \good\in \goods,\;
    \sum_{\buyer\in \buyers} \allocationstar[\buyer][\good]-1\leq 0$, $-\allocationstar[\buyer][\good]\leq 0$.
    \item (Dual Feasibility) $\forall \good\in \goods,\;\pricelangstar[\buyer][\good]\geq 0$, $\bmustar[\buyer][\good]\geq 0$.
\end{itemize}
Then, since $\allocationstar$ is a VE of the jointly-convex pseudo-game $\pgame$, by \Cref{thm:jointly_convex_kkt}, there exists a common optimal multiplier $\pricelangstar[1]=\hdots=\pricelangstar[\numbuyers]=\pricestar$ such that for each buyer $\buyer$, $(\allocationstar[\buyer], \pricestar)$ satisfies her KKT conditions. We will show that $(\allocationstar, \pricestar)$ is a competitive equilibrium of the influence Fisher market $(\graph, \util, \budget)$. 

From the complementary slackness and primal feasibility conditions, we know that 
\begin{align*}
    \sum_{\good\in \goods} \pricestar[\good](\sum_{\buyer\in \buyers}\allocationstar[\buyer][\good]-1)=0\\
    \forall \good\in \goods,\; \sum_{\buyer\in \buyers} \allocationstar[\buyer][\good]\leq 1
\end{align*}
Thus, $(\allocationstar, \pricestar)$ satisfies the market clearance. 

Moreover, from the stationarity conditions, we get for any $\buyer\in \buyers$:
\begin{align*}
    \frac{\partial \lang[\buyer]}{\partial \allocation[\buyer][\good]}
    &= \frac{-\budget[\buyer]}{\util[\buyer](\allocationstar[\buyer], \allocationstar[\nei])}\left[
    \frac{\partial \util[\buyer]}{\partial \allocation[\buyer][\good]} 
    \right]_{\allocation[\buyer]=\allocationstar[\buyer]}
    + \pricestar[\good] - \bmustar[\buyer][\good]
    = 0\\
    \pricestar[\good] 
    &=  \frac{\budget[\buyer]}{\util[\buyer](\allocationstar[\buyer], \allocationstar[\nei])}\left[
    \frac{\partial \util[\buyer]}{\partial \allocation[\buyer][\good]} 
    \right]_{\allocation[\buyer]=\allocationstar[\buyer]} - \bmustar[\buyer][\good]\\
    \pricestar[\good]\allocationstar[\buyer][\good] 
    &=  \frac{\budget[\buyer]}{\util[\buyer](\allocationstar[\buyer], \allocationstar[\nei])}\left[
    \frac{\partial \util[\buyer]}{\partial \allocation[\buyer][\good]} 
    \right]_{\allocation[\buyer]=\allocationstar[\buyer]} \allocationstar[\buyer][\good] - \bmustar[\buyer][\good] \allocationstar[\buyer][\good]
\\
    \sum_{\good\in \goods}\pricestar[\good]\allocationstar[\buyer][\good] 
    &=  \frac{\budget[\buyer]}{\util[\buyer](\allocationstar[\buyer], \allocationstar[\nei])}
    \sum_{\good\in \goods}\left[
    \frac{\partial \util[\buyer]}{\partial \allocation[\buyer][\good]} 
    \right]_{\allocation[\buyer]=\allocationstar[\buyer]} \allocationstar[\buyer][\good]\\
    \sum_{\good\in \goods}\pricestar[\good]\allocationstar[\buyer][\good]
    &=   \frac{\budget[\buyer]}{\util[\buyer](\allocationstar[\buyer], \allocationstar[\nei])} \util[\buyer](\allocationstar[\buyer], \allocationstar[\nei])
\\
    \sum_{\good\in \goods}\pricestar[\good]\allocationstar[\buyer][\good]
    &= \budget[\buyer]
\end{align*}
where the third line is from multiplying both sides by $\allocation[\buyer][\good]$, and the fifth line is from the Euler's Theorem for homogeneous functions. 
Note that the left hand side of this expression is exactly the spending of buyer $\buyer$ at $(\allocationstar, \pricestar)$. This results implies that consumers are not spending more than their budgets. 

Finally, we want to show that $(\allocationstar, \pricestar)$ is utility maximizing. From the stationarity conditions again, we get for any $\buyer\in \buyers$:
\begin{align}
     \frac{\partial \lang[\buyer]}{\partial \allocation[\buyer][\good]}
    &= \frac{-\budget[\buyer]}{\util[\buyer](\allocationstar[\buyer], \allocationstar[\nei])}\left[
    \frac{\partial \util[\buyer]}{\partial \allocation[\buyer][\good]} 
    \right]_{\allocation[\buyer]=\allocationstar[\buyer]}
    + \pricestar[\good] - \bmustar[\buyer][\good]
    = 0\\
\end{align}
If $\allocationstar[\buyer][\good]>0$, by complementary slackness condition, $\bmustar[\buyer][\good]=0$, which gives us
\begin{align*}
    \pricestar[\good] 
    &=  \frac{\budget[\buyer]}{\util[\buyer](\allocationstar[\buyer], \allocationstar[\nei])}
    \left[
    \frac{\partial \util[\buyer]}{\partial \allocation[\buyer][\good]} 
    \right]_{\allocation[\buyer]=\allocationstar[\buyer]}\\
    \frac{\util[\buyer](\allocationstar[\buyer], \allocationstar[\nei]) }{\budget[\buyer]}
    &= \frac{ \left[
    \frac{\partial \util[\buyer]}{\partial \allocation[\buyer][\good]} 
    \right]_{\allocation[\buyer]=\allocationstar[\buyer]}}{\pricestar[\good]}
\end{align*}
The last condition is exactly the equimarginal principle \citep{mas-colell}, hence $(\allocationstar, \pricestar)$ is utility maximizing.

Therefore, any VE $\allocationstar$ of $\pgame$ constitutes an equilibrium allocation of $(\graph, \util, \budget)$, and the optimal Lagrangian multiplier that corresponds to the joint constraint are the corresponding equilibrium prices. 
\end{proof}

\thmPseudoConvergence*
\begin{proof}
Consider the variational inequality problem $\vi(\actions, \objs)$ for the jointly convex pseudo-game $\pgame$ defined in \Cref{thm:pseudo_game_equ} with $\actions=\{\allocation\in \R^{\numbuyers\times \numgoods}\mid \ones - \sum_{\buyer\in \buyers} \allocation[\buyer]\geq \zeros\}$ and $\objs(\allocation)\coloneqq (\grad[{\allocation[\buyer]}]\utilp[\buyer](\allocation))_{\player\in \players}$. Since each $\util[\buyer]$ is jointly-concave in every buyer's allocation $\allocation[\buyer]$, $\utilp[\buyer]$ defined by $\utilp[\buyer](\allocation)=\budget[\buyer]\log(\util[\buyer](\allocation[\buyer], \allocation[\nei]))$ is jointly-concave in every buyer's allocation by the composition of scalar functions as $\log$ is concave. Thus, the operator $\objs$ is monotone. 
Moreover, since each $\util[\buyer]$ is twice-differentiable, the operator $\objs$ is $\lipcont[\objs]$-Lipschitz where $\lipcont[\objs]=\max_{\allocation\in \Rp^{\numbuyers\cdot \numgoods}: \allocation\in \actions}\|\grad[\allocation]\objs(\allocation)\|$. Note that though $\objs$ is not differentiable at $\allocation=\zeros$ as $\util[\buyer](\zeros)=0$ for all $\buyer$, we can remedy that by shifting up all $\util[\buyer]$ by a small constant $\varepsilon$ albeit losing some accuracy. 

Then, given this monotone, Lipschitz variational inequality, 
Extragradient Method (EG) can converge  in last-iterate at rate of $O(\nicefrac{1}{T})$ \cite{gorbunov2022extragradient}. 
\end{proof}

\begin{lemma}\label{lemma:opt_lambda}
The optimization problem 
\begin{align}
    \max_{\allocation[\buyer] \in \Rp^{\numgoods}: \allocation[\buyer]\cdot \price \leq \budget[\buyer]}  \budget[\buyer] \log(\util[\buyer](\allocation[\buyer], \allocation[\nei]))
\end{align}
is equivalent to the optimization problem 
\begin{align}
    \max_{\allocation[\buyer] \in \Rp^{\numgoods}}  \budget[\buyer] \log(\util[\buyer](\allocation[\buyer], \allocation[\nei])) + \budget[\buyer] - \allocation[\buyer]\cdot \price
\end{align}
\end{lemma}
\begin{proof}
The Lagrangian associated with $\max_{\allocation[\buyer] \in \Rp^{\numgoods}: \allocation[\buyer]\cdot \price \leq \budget[\buyer]}  \budget[\buyer] \log(\util[\buyer](\allocation[\buyer], \allocationstar[\nei]))$ is given by
\begin{align*}
    \lang(\allocation[\buyer], \langmult[ ], \bmu)
    = \budget[\buyer] \log(\util[\buyer](\allocation[\buyer], \allocation[\nei]))
    + \langmult[ ](\budget[\buyer]-\allocation[\buyer]\cdot\price)
    + \bmu^T\allocation[\buyer]
\end{align*}
where $\langmult[ ]\in \Rp$ and $\bmu\in \Rp^{\numgoods}$ are slack variables.

Let $(\allocationstar[\buyer], \langmultstar[ ], \bmustar)$ be an optimal solution to the Lagrangian. From the KKT stationary condition for this Lagrangian, it holds that, for all $\good\in \goods$,
\begin{align*}
    \frac{\budget[\buyer]}{\util[\buyer](\allocationstar[\buyer], \allocation[\nei])} 
    \left[ \frac{\partial \util[\buyer]}{\partial \allocation[\buyer][\good]}
    \right]_{\allocation[\buyer]=\allocationstar[\buyer]}
    -\langmultstar[ ]\price[\good] + \bmustar[\good] = 0\\
    \frac{\budget[\buyer]}{\util[\buyer](\allocationstar[\buyer], \allocation[\nei])} 
    \left[ \frac{\partial \util[\buyer]}{\partial \allocation[\buyer][\good]}
    \right]_{\allocation[\buyer]=\allocationstar[\buyer]} \allocationstar[\buyer][\good]
    -\langmultstar[ ]\price[\good]\allocationstar[\buyer][\good] + \bmustar[\good]\allocationstar[\buyer][\good] = 0\\
      \frac{\budget[\buyer]}{\util[\buyer](\allocationstar[\buyer], \allocation[\nei])} 
    \left[ \frac{\partial \util[\buyer]}{\partial \allocation[\buyer][\good]}
    \right]_{\allocation[\buyer]=\allocationstar[\buyer]} \allocationstar[\buyer][\good]
    -\langmultstar[ ]\price[\good]\allocationstar[\buyer][\good]  = 0,
\end{align*}
where the second line is obtained by multiplying both sides by $\allocationstar[\buyer][\good]$, and the third line is obtained by the KKT complementary condition, i.e., $\forall \good\in \goods, \bmustar[\good]\allocationstar[\buyer][\good]=0$ .

Summing up across all $\good\in \goods$ on both sides yields:
\begin{align*}
    \frac{\budget[\buyer]}{\util[\buyer](\allocationstar[\buyer], \allocation[\nei])} 
    \sum_{\good\in \goods}\left[ \frac{\partial \util[\buyer]}{\partial \allocation[\buyer][\good]}
    \right]_{\allocation[\buyer]=\allocationstar[\buyer]}  \allocationstar[\buyer][\good]
    - \langmultstar[ ]\sum_{\good\in \goods} \price[\good]\allocationstar[\buyer][\good]
    =0\\
    \frac{\budget[\buyer]}{\util[\buyer](\allocationstar[\buyer], \allocation[\nei])} \util[\buyer](\allocationstar[\buyer], \allocation[\nei])
    - \langmultstar[ ]\sum_{\good\in \goods} \price[\good]\allocationstar[\buyer][\good]
    =0\\
    \budget[\buyer] - \langmultstar[ ]\budget[\buyer] = 0\\
    \langmultstar[ ]= 1,
\end{align*}
where the second line is obtained from the Euler's theorem for homogeneous functions, and the last line is from the KKT complementary condition again, i.e., $\langmultstar[ ](\budget[\buyer]-\allocationstar[\buyer]\cdot\price)= \langmultstar[ ](
\sum_{\good\in \goods}\budget[\buyer]-\price[\good]\allocationstar[\buyer][\good])=0$. 

Hence, plugging $\langmultstar[ ]=1$ back into the Lagrangain restritced to $\Rp^{\numgoods}$, we get:
\begin{align*}
        &\max_{\allocation[\buyer] \in \Rp^{\numgoods}: \allocation[\buyer]\cdot \price \leq \budget[\buyer]}  \budget[\buyer] \log(\util[\buyer](\allocation[\buyer], \allocation[\nei]))\\
        =     &\max_{\allocation[\buyer] \in \Rp^{\numgoods}}  \budget[\buyer] \log(\util[\buyer](\allocation[\buyer], \allocation[\nei])) + \langmultstar[ ](\budget[\buyer] - \allocation[\buyer]\cdot \price)\\
        =     &\max_{\allocation[\buyer] \in \Rp^{\numgoods}}  \budget[\buyer] \log(\util[\buyer](\allocation[\buyer], \allocation[\nei])) + \budget[\buyer] - \allocation[\buyer]\cdot \price
\end{align*}
\end{proof}

\thmPseudoDual*
\begin{proof}
Since for each $\buyer\in \buyers$, $\pricestar$ is the optimal Lagrangian multiplier, we can characterize the price $\pricestar$ through the Lagrangian dual function:
\begin{align*}
    g_{\buyer}(\price) 
    &= \max_{\allocation[\buyer]\in \Rp^{\numgoods}} \lang[\buyer](\allocation[\buyer], \allocationstar[\nei], \price)\\
    &= \max_{\allocation[\buyer]\in \Rp^{\numgoods}}
    \bigg\{\budget[\buyer]\log(\util[\buyer](\allocation[\buyer], \allocationstar[\nei])) \\
    &+ \sum_{\good\in \goods} \price[\good] (1-\allocation[\buyer][\good] - \sum_{k\neq \buyer} \allocationstar[k][\good])
    \bigg\}\\
    &= \sum_{\good\in \goods}\price[\good] - \sum_{k\neq i}\sum_{\good\in \goods} \price[\good]\allocationstar[k][\good]\\
    &+ \max_{\allocation[\buyer]\in \Rp^{\numgoods}}
    \left\{\budget[\buyer]\log(\util[\buyer](\allocation[\buyer], \allocationstar[\nei])) 
    + \sum_{\good\in \goods}\price[\good]\allocation[\buyer][\good] 
    \right\}\\
    &= \sum_{\good\in \goods}\price[\good] - \sum_{k\neq i}\sum_{\good\in \goods} \price[\good]\allocationstar[k][\good]\\
    &+ \max_{\allocation[\buyer]\in \Rp^{\numgoods}}
    \left\{\budget[\buyer]\log(\util[\buyer](\allocation[\buyer], \allocationstar[\nei])) 
    + \budget[\buyer]-\price\cdot \allocation[\buyer]\right\}-\budget[\buyer]\\
    &=\sum_{\good\in \goods}\price[\good] - \sum_{k\neq i}\sum_{\good\in \goods} \price[\good]\allocationstar[k][\good]\\
    &+\max_{\allocation[\buyer] \in \Rp^{\numgoods}: \allocation[\buyer]\cdot \price \leq \budget[\buyer]}  \budget[\buyer] \log(\util[\buyer](\allocation[\buyer], \allocationstar[\nei])) - \budget[\buyer]
\end{align*}
where the fourth line is from \Cref{lemma:opt_lambda}. Therefore, the dual of optimization problem for buyer $\buyer$ is $\min_{\price\in \Rp^{\numgoods}} g_{\buyer}(\price)
=\min_{\price\in \Rp^{\numgoods}}
    \sum_{\good\in \goods}\price[\good] - \sum_{k\neq i}\sum_{\good\in \goods} \price[\good]\allocationstar[k][\good]
    +\max_{\allocation[\buyer] \in \Rp^{\numgoods}: \allocation[\buyer]\cdot \price \leq \budget[\buyer]}  \budget[\buyer] \log(\util[\buyer](\allocation[\buyer], \allocationstar[\nei])) - \budget[\buyer]$ .
\end{proof}




\begin{assumption} \label{assumption:envelope_thm}
1. $\obj$, $\constr[1],\hdots, \constr[K]$ are continuous and concave in $\inner$; 2. $\grad[\outer]\obj$, $\grad[\outer]\constr[1],\hdots, \grad[\outer]\constr[K]$ are continuous in $(\outer, \inner)$; and $\forall \outer \in \outerset$, $\exists \innerp\in \innerset$ s.t. $\constr[k](\outer, \innerp)>0$ for all $k=1,\hdots, K.$
\end{assumption}

\begin{lemma}[Subdifferential Envelope Theorem]\cite{goktas2021minmax}
\label{lemma:subdiff_envelop}
Consider the value function $\val(\outer)=\max_{\inner\in \innerset:\constr(\outer,\inner)\geq \zeros} \obj(\outer,\inner)$. Let $Y^*(\outer)=\argmax_{\inner\in \innerset: \constr(\outer,\inner)\geq \zeros }\obj(\outer,\inner)$ and suppose \Cref{assumption:envelope_thm} holds. Then, at any point $\outerp\in \outerset$, $\subdiff[\outer] \val(\outer)=$
\begin{align*}
    \text{conv}\left(
    \bigcup_{\inner^*(\outerp)\in \innerset^*(\outerp)}
    \bigcup_{\lambda_k(\outerp, \inner^*(\outerp))\in \Lambda(\outerp, \inner^*(\outerp))}
    \left\{
    \grad[\outer] \obj(\outerp, \inner^*(\outerp)) \right. \right.\\
   \left.\left. + \sum_{k=1}^K \lambda_k(\outerp,\inner^*(\outerp)) \grad[\outer] \constr[k](\outerp, \inner^*(\outerp))
    \right\}
    \right),
\end{align*}
where $\subdiff$ is the subdifferential operator, $\bm{\lambda}(\outerp, \inner^*(\outerp))=(\lambda_1(\outerp, \inner^*(\outerp), \hdots, \lambda_K(\outerp, \inner^*(\outerp)) \in \Lambda(\outerp, \inner^*(\outerp)$ are the Langrange multipliers associated with $\inner^*(\outerp)\in Y^*(\outerp)$,
and $\text{conv}$ is the convex hull operator. 
\end{lemma}


\thmSubdiff*
\begin{proof}
For all goods $\good\in \goods$,
\begin{align}
    \subdiff[{\price[\good]}]
    \val[\buyer](\price)
    &= \subdiff[{\price[\good]}] \left(
 \max_{\allocation[\buyer] \in \Rp^{\numgoods}: \allocation[\buyer] \cdot \price \leq \budget[\buyer]}
    \sum_{\good\in \goods} \price[\good] \left( 1 - \sum_{k \neq \buyer} \allocationstar[k][\good] \right)  \right.\nonumber
   \\
   &\left.+ \budget[\buyer] \log( \util[\buyer] (\allocation[\buyer], \allocationstar[\nei]))    
   \right)\\
    &= \subdiff[{\price[\good]}] \left(
    \sum_{\good\in \goods} \price[\good] \left( 1 - \sum_{k \neq \buyer} \allocationstar[k][\good] \right) \right.  \nonumber\\
    &\left.+  \max_{\allocation[\buyer] \in \Rp^{\numgoods}: \allocation[\buyer] \cdot \price \leq \budget[\buyer]} \budget[\buyer] \log( \util[\buyer] (\allocation[\buyer], \allocationstar[\nei]))    
   \right)\\
    &= \subdiff[{\price[\good]}] \left(
    \sum_{\good\in \goods} \price[\good] \left( 1 - \sum_{k \neq \buyer} \allocationstar[k][\good] \right) \right) \nonumber\\
    & + \subdiff[{\price[\good]}]\left(\max_{\allocation[\buyer] \in \Rp^{\numgoods}: \allocation[\buyer] \cdot \price \leq \budget[\buyer]} \budget[\buyer] \log( \util[\buyer] (\allocation[\buyer], \allocationstar[\nei])) \right)\\
    &= \left( 1 - \sum_{k \neq \buyer} \allocationstar[k][\good] \right) 
    \nonumber \\
    &+ \subdiff[{\price[\good]}]\left(\max_{\allocation[\buyer] \in \Rp^{\numgoods}} \budget[\buyer] \log( \util[\buyer] (\allocation[\buyer], \allocationstar[\nei])) + \budget[\buyer] - \allocation[\buyer]\cdot \price \right) \label{eq:use_opt_lambda}\\
    &= 1 - \sum_{k \neq \buyer} \allocationstar[k][\good] + \allocationstar[\buyer] \label{eq:use_subdiff}\\
    &= 1 - \sum_{\buyer\in \buyers} \allocationstar[k][\good]
\end{align}
where \cref{eq:use_opt_lambda} is from \Cref{lemma:opt_lambda}, and \cref{eq:use_subdiff} is from \Cref{lemma:subdiff_envelop}. 
\end{proof}

\thmStackelbergConvergence*
\begin{proof}
First, note that for each buyer $\buyer$'s optimization problem, the dual given by (\cref{eq:individual_dual}) is convex in $\price$. Thus, the value function $\val$ is also convex in $\price$ as it is the sum of the duals. 
Moreover, we know  $\grad[\price]\val(\price)=\ones-\sum_{\buyer\in \buyers}\allocationstar[\buyer]$  from \Cref{thm:subdiff_equal_excess_demands}. Note that when $\price>\zeros$, $\allocationstar$ is bounded as $\allocationstar[\buyer]\leq \nicefrac{\budget[\buyer]}{\price}$ for each $\buyer\in \buyers$, so $\grad[\price]\val(\price)$ is also bounded. 
 \ssadie{}{Note that though $\allocationstar[\buyer]$ is not necessarily bounded at $\price=0$, we can remedy that by shifting up $\price$ by a small constant $\varepsilon$ albeit losing some accuracy. }
Therefore,  $\val$ is $\lipcont[\val]$-Lipschitz, where $\lipcont[\val]=\max_{\price\in \Rpp^{\numgoods}}\|\grad[\price]\val(\price)\|$ for $\price\in \Rpp$. 
Then, using a subgradient method, the algorithm converges in average-iterate at a rate of $O(\nicefrac{1}{\sqrt{T}})$.
\end{proof}

\section{Experimental Setup}
The main goal of our experiment is to understand the empirical convergence rate of \Cref{alg:ne_oracle_gd} in different Fisher markets, in which the objective function in \Cref{eq:obj_func} satisfies different smoothness properties. To answer the question, we ran multiple experiments, each time recording the prices and allocations computed by \Cref{alg:nested_ne_gd} during each iteration $\iterouter$ of the main (outer) loop. For each run of each algorithm on each market with each set of initial conditions, we then computed the objective function’s value for the iterates, i.e., $\obj(\allocation^{(\iterouter)}, \price^{(\iterouter)})$, which we plot in \Cref{fig:convergence_results}.

\paragraph{Hyperparameters}
We randomly initialized 50 different linear, Cobb-Douglas, Leontief Fisher markets with social influence, each with 3 buyers and 3 goods. Buyer $\buyer$’s budget $\budget[\buyer]$ was drawn randomly from a uniform distribution ranging from 5 to 15 (i.e., $U[5,15]$), while each buyer $\buyer$’s valuation for good $\good$, $\valuation[\buyer][\good]$, was drawn randomly from $U[5,35]$. The social network graph $\graph$ is also generated uniformly at random.

For influence Fisher markets with linear utilities, we ran our algorithm for 400 iterations with learning rate $\learnrate[\price]=2$, and solving the inner Nash equilibrium by running the extragradient method for 100 iterations with learning rate $\learnrate[\allocation]=0.2$.
For influence Fisher markets with Cobb-Douglas utilities, we ran our algorithm for 400 iterations with learning rate $\learnrate[\price]=8$, and solving the inner Nash equilibrium by running the extragradient method for 200 iterations with learning rate $\learnrate[\allocation]=0.5$.
Finally, for influence Fisher markets with Leontief utilities, we ran our algorithm for 400 iterations with learning rate $\learnrate[\price]=5$, and solving the inner Nash equilibrium by running the extragradient method for 100 iterations with learning rate $\learnrate[\allocation]=3$.

\paragraph{Programming Languages, Packages, and Licensing}
We ran our experiments in Python 3.7, using NumPy.
\Cref{fig:convergence_results} were graphed using Matplotlib.

Python software and documentation are licensed under the PSF License Agreement. Numpy is distributed under a liberal BSD license.  Matplotlib only uses BSD compatible code, and its license is based on the PSF license. CVXPY is licensed under an APACHE license. 

\paragraph{Computational Resources}
Our experiments were run on Google Colab with 12.68GB RAM, and took about 3 hours to run experiments with 50 markets. 

\paragraph{Code Repository}
The data our experiments generated, and the code used to produce our visualizations, can be found in our code repository ({\color{blue}\rawcoderepo}).

\end{document}